\documentclass[12pt,authoryear,sort]{article}

\usepackage{hyperref}
\usepackage{amsmath,amsfonts,amssymb,graphicx,booktabs, authblk}
\usepackage{algorithm, algpseudocode, mathbbol}
\usepackage{url}
\usepackage{tikz}
\usepackage{amsthm, amssymb}
\usetikzlibrary{shapes.geometric, arrows}
\usepackage{pgfplotstable}
\usepackage{pgfplots}
\usepackage[margin =1in]{geometry}
\usepackage{subcaption}
\usepackage[mode=buildnew]{standalone}
\usepackage{cellspace}
\newcommand{\old}{\operatorname{old}}
\usepackage[skip=-1pt]{caption}

\def\myalgfont{\small}
\algrenewcommand\textproc{\myalgfont}
\algrenewcommand\algorithmicprocedure{{\myalgfont \bf procedure}}

\usepackage{natbib}

\newcommand{\VaR}{\operatorname{VaR}}
\newcommand{\tr}{^\intercal}
\newcommand{\xirnd}{\xi \hspace{-.440em} \xi}
\newcommand\R{{\sf I\kern-0.1em R}}
\newcommand{\st}{{\rm s.t.}}
 \newcommand\one{e}
 \newcommand{\X}{\mathcal{X}}
 \newcommand{\bmin}{\underline{b}}
 \newcommand{\argmin}{\operatorname{argmin}}

\newtheorem{theorem}{Theorem}
\newtheorem{lemma}{Lemma}
\newtheorem{proposition}{Proposition}
\newtheorem{corollary}{Corollary}
\newtheorem{example}{Example}
\newtheorem{remark}{Remark}

\tikzstyle{startstop} = [rectangle, rounded corners, minimum width=3cm, minimum height=1cm,text centered, text width=3cm, draw=black, fill=red!30]

\tikzstyle{io} = [trapezium, trapezium left angle=70, trapezium right angle=110, minimum width=3cm, minimum height=1cm, text centered, text width=3cm, draw=black, fill=blue!30]

\tikzstyle{process} = [rectangle, minimum width=3cm, minimum height=1cm, text centered, text width=4cm, draw=black, fill=orange!30]

\tikzstyle{process2} = [rectangle, minimum width=3cm, minimum height=1cm, text centered, text width=3cm, draw=black, fill=orange!30]

\tikzstyle{decision} = [diamond, minimum width=3cm, minimum height=1cm, text centered, text width=3cm, draw=black, fill=green!30]

\tikzstyle{arrow} = [thick,->,>=stealth]

\begin{document}


\title{Computing near-optimal Value-at-Risk portfolios using Integer Programming techniques}

\author[1]{Onur Babat \thanks{\tt onur.babat@lehigh.edu}}
\author[2]{Juan ~C. Vera\thanks{\tt j.c.veralizcano@uvt.nl}}
\author[3]{Luis ~F. Zuluaga\thanks{\tt luis.zuluaga@lehigh.edu}}

\affil[1]{Department of Industrial and Systems Engineering, Lehigh University, Bethlehem,
PA USA}
\affil[2]{Department of Econometrics and Operations Research, Tilburg University, Tilburg,
 The Netherlands}
\affil[3]{Department of Industrial and Systems Engineering, Lehigh University, Bethlehem,
PA USA}

\date{}

\maketitle

\begin{abstract}
Value-at-Risk (VaR)
is
one of the main regulatory tools used  for risk
management purposes.
However,
it is difficult to compute optimal VaR portfolios; that is,  an optimal risk-reward portfolio
allocation
using VaR as the risk measure. This is due to
VaR being non-convex and
of combinatorial nature.
In particular, it is well-known that the VaR portfolio problem can be formulated
as a mixed-integer linear program (MILP) that is difficult to solve with current MILP solvers for
medium to large-scale instances of the problem.
 Here, we present an algorithm to compute
near-optimal VaR portfolios that takes advantage of this MILP formulation and
provides a guarantee of the solution's near-optimality. As a byproduct, we obtain an algorithm
to compute tight lower bounds on the VaR portfolio problem that outperform related
algorithms proposed in the literature for this purpose. The near-optimality guarantee provided
by the proposed algorithm is obtained
thanks to the relation
between minimum risk portfolios satisfying a reward benchmark and the
corresponding maximum reward portfolios satisfying a risk benchmark. These
alternate formulations of the portfolio allocation problem have been frequently studied in
the case of
convex risk measures and concave
reward functions. Here, this relationship is considered for general risk measures and reward functions.
To illustrate the efficiency of the presented algorithm, numerical results are
presented using historical asset returns from the US financial market.
\end{abstract}

{\bf Keywords:} Value-at-Risk; Portfolio Allocation; Integer Programming Relaxations; Minimizing Risk vs Maximizing Reward equivalence


\section{Introduction}
In the context of portfolio risk and asset liability management, Value-at-Risk (VaR)
measures the exposure of a portfolio to high losses.
VaR is prominent in current regulatory frameworks for banks (see, e.g., the Basel II and Basel III Accords),
as well as for insurance companies (see, e.g., the Solvency II Directive).
Thus, VaR is an important  and popular tool for risk management in the modern financial
and risk management literature \citep[see, e.g.,][]{VaRbook,Wozabal12}.
Accordingly, the development of risk management methods based on
VaR has been the focus of extensive research work  \citep[see, e.g.,][]{Coleman, Basak, Darbha, Alexei, Koll, Ghaoui, Glass,Wozabal10, Benati07, Gnei11}. 

Although VaR is widely used to measure the risk of a given portfolio
of assets, it is not commonly used as a risk measure in the context
of computing optimal VaR portfolios; that is,  an optimal risk-reward portfolio
allocation
using VaR as the risk measure. Instead, other risk measures
such as the portfolio return's Variance  \citep[cf.,][]{Markowitz},
and the portfolio loss' Conditional Value-at-Risk (CVaR) \citep[cf.,][]{Uryasev00} are more commonly used. This is because, in contrast with the above
mentioned risk measures, VaR is non-convex and
of combinatorial nature \citep[cf.,][]{Alexei}. As a result, the VaR portfolio problem is inherently
difficult solve~\citep[see, e.g.,][]{SimPN09}.

VaR does not (in general) satisfies the commonly accepted axioms
of {\em coherent} risk measures~\citep[cf.,][]{CoherentII,Coherent}).
On the other hand, VaR satisfies the so-called {\em natural risk statistic} axioms \citep{nocoherent}.
More importantly, it has been recently shown in \citet{Gneiting11} that VaR
is an {\em elicitable} risk measure~\citep[cf.,][]{Bellini13}. Loosely speaking,
elicitability is related to how
accurately a risk measure can be forecasted. More precisely, as discussed
in \citet{Bellini13}, it has
been shown that while
CVaR is generally considered a better risk measure from a mathematical point of view, it
requires a higher number of samples for an accurate estimation \citep[see][]{Dani2011} and it is
less robust than VaR \citep[see][]{Cont10}.

Because of the computational difficulties of optimally solving general instances
of the VaR portfolio problem, different heuristics have been proposed
in the literature. In particular, consider the
work of  \citet{Alexei, LarsMU02, Verma, Thiele14}. Also, given that the VaR portfolio
problem belongs to the general class of {\em chance constraint} optimization
problems \citep[cf.,][]{CampC05}, other approximation approaches that can be
used are based on relaxations of the VaR quantile constraint
for which probabilistic guarantees can be obtained.
In particular, consider the work of \citet{CampC05, Farias04, Erdo06}.

When the standard {\em sampling approach} \citep[cf.,][]{Uryasev00} is used to model the
uncertain asset returns, it is well known \citep[see, e.g.,][]{Benati07, FengWS15} that the (resulting) VaR portfolio problem
can be solved to optimality by formulating the problem as a
mixed-integer linear program (MILP).
However, this formulation is difficult to solve with current MILP solvers for
instances with medium to large
number of assets in the portfolio \citep[see, e.g.,][]{Benati07}. Recently, improvements
in the solution of this MILP formulation have been obtained in
\citet{FengWS15}, by tailoring special {\em branch-and-cut} techniques
to solve the problem, as well as improving the {\em big-M} values used
on its MILP formulation. Although these improvements allow for the
solution of VaR portfolio problem instances where thousands of scenarios are used
to model the uncertain asset returns, the number of assets that are considered
in the portfolio is of the order of 25 assets, similar to \citet{Benati07}. Moreover, their solution approach is useful
only when the common total wealth constraint is not considered \citep[][Sec. 5]{FengWS15}.

We present an algorithm to compute
near-optimal VaR portfolios that takes advantage of the VaR portfolio problem MILP formulation and
provides a guarantee of the near-optimality
of the solution. The algorithm makes a straight-forward use of current state-of-the-art
MILP solvers (e.g., {\tt CPLEX} and {\tt Gurobi}). Furthermore,
this algorithm can be used to obtain guaranteed near-optimal solutions
for instances of the VaR problem with up to a hundred assets and thousands of samples to model the
uncertain asset returns. In particular, this allows the use of VaR for strategic asset allocation instead of only tactical asset allocation (e.g., by industry sectors). As a byproduct, we obtain an algorithm
to compute tight lower bounds on the VaR portfolio problem that outperforms the  algorithms for this purpose
recently
proposed by~\cite{LarsMU02}. These algorithms aim at approximating the optimal solution of the VaR portfolio problem by iteratively constructing appropriate instances of the Conditional Value-at-Risk portfolio problem.

The main contribution of the article in relation to the current
VaR portfolio allocation literature is to provide a performance-guaranteed heuristic
solution approach for the problem which can be used to address the solution of medium to large-scale
instances of the problem.
The near-optimal guarantee provided by the proposed algorithm is based on the relation between
two alternate formulations of the portfolio problem; that is,
between minimum risk portfolios satisfying a reward benchmark and the
corresponding maximum reward portfolios satisfying a risk benchmark.
It is well-known that these alternate formulations of the portfolio
problem are equivalent for the mean-variance portfolio model of \citet{Markowitz}.
Recently, \citet[][Thm. 3]{Krok02} have shown
that this equivalence holds for general convex risk measures and concave
reward functions. We also study the relationship between the
alternate  risk-reward and reward-risk formulations of the portfolio problem
for more general risk measures and reward functions. Besides providing
the foundation for the proposed algorithm
to find near-optimal solutions for the
VaR portfolio problem, these results extend the characterization provided by \citet[][Thm. 3]{Krok02}, and rectify some incorrect statements
made in \citet{Lin09} about alternate formulations of the
VaR portfolio problem.

The rest of the article is organized as follows. In Section~\ref{MILPformulation},
the MILP formulation of the VaR portfolio problem is presented. In Section~\ref{S:Dual},
the relationship between the alternate formulations of the portfolio
problem is studied for general risk measures and reward functions. These results
are used in Section~\ref{algorithm}
to develop the proposed algorithm to find near-optimal solutions for the
VaR portfolio problem. In Section~\ref{numerical},
we illustrate the efficiency of the proposed algorithm
by presenting numerical results on instances of the VaR portfolio
problem constructed using historical asset returns from the US financial market.

\section{The MILP formulation of the VaR portfolio problem}
\label{MILPformulation}
The Value-at-Risk (VaR) of a portfolio measures its exposure to high losses.
Specifically,
for a given $\alpha \in (0,1)$ (typically $0.01 \le \alpha \le 0.05$), the
VaR of a portfolio is defined as the $1-\alpha$ quantile of the
portfolio's losses \citep[cf.,][]{Uryasev00}; or equivalently as the $\alpha$ quantile of the
portfolio's returns. Here,
we follow the latter definition~\citep[as in][]{Alexei}.

We begin by formally
stating the VaR portfolio (allocation) problem; which aims at minimizing the
exposure of the portfolio to high losses while maintaining a minimum expected
level of profit.
Consider $n$ risky assets that can be chosen by an investor in the
financial market. Let $\xirnd = {(\xirnd_1,\dots,\xirnd_n)}\tr$ be a random variable in $\R^n$ representing the uncertain
returns of the $n$ risky assets from the current time $t=0$ to a fixed future
time $t=T$. Let $x={(x_1,\dots,x_n)}\tr \in \R^n_+$ denote a portfolio on these assets;
that is, the percentage of the
available funds to be allocated in each of the $n$ risky assets.
Following~\citet{Alexei}, given $\alpha \in (0,1)$, the $\alpha$-level VaR of the portfolio is defined as follows:
\begin{equation}
\label{vardef}
\VaR_{\alpha}(x{\tr}\xirnd) = -\mathbb{Q}_{\alpha}(x{\tr}\xirnd),
\end{equation}
where $\mathbb{Q}_{\alpha}(x{\tr}\xirnd)$ is the $\alpha$ quantile of the portfolio's return
distribution;
that is, $\mathbb{Q}_{\alpha}(x{\tr}\xirnd)= \inf\{v: \mathbb{P}{\rm r}(x{\tr}\xirnd \le v) > \alpha\}$. Also, the
expected portfolio return from $t=0$ to $t=T$ is given by~$\mathbb{E}(x{\tr}\xirnd)$.
Above, $\mathbb{P}{\rm r}(\cdot)$ and $\mathbb{E}(\cdot)$ respectively indicate
probability and expectation.

A (single-period) VaR
portfolio  problem aims at finding the portfolio~$x\in \R^n_+$
to be constructed at $t=0$,
in order to minimize the portfolio's $\VaR_{\alpha}$, subject to the portfolio having a given minimum expected return
$\mu_0$, and possibly satisfying some linear diversification constraints. Formally,
the VaR portfolio problem is:
\begin{equation}
\label{nominal}
\begin{array}{lll}
\min & -\VaR_{\alpha}(x{\tr}\xirnd)  \\[2ex]
\st  & \mathbb{E}(x{\tr}\xirnd) \ge \mu_0\\
     & x{\tr}\one = 1\\
     & x \in \mathcal{X} \subseteq \R^n_+,\\
     \end{array}
\end{equation}
where $\one \in \R^n$ is the vector of all-ones, $\mu_0 \in \R$ is the given target minimum expected portfolio return, and $\mathcal{X} \subseteq \R^n$ is a given set
defined by linear constraints; which are typically used to enforce certain
diversification constraints on the portfolio $x \in \R^n_+$. For the moment, it is assumed that
no short positions are allowed in the portfolio; which
is the most common situation in practice \citep[cf.,][]{Michaud}.

As discussed in \citet{Alexei}, there are two main approaches to solve~\eqref{nominal}: the
{\em parametric} approach, in which it is assumed that the asset returns are governed by a known
distribution (see, e.g., \citet{Lobo}, where asset returns are assumed to be normally distributed);
and the {\em sampling} approach, which uses a finite number of samples $\xi^1, \dots, \xi^m \in \R^n$ of
the asset returns (see, e.g., \citet{Alexei}), that are typically obtained from historical data,
simulations, or a
combination of both. The latter approach
is used in the well-known Conditional Value-at-Risk (CVaR) portfolio model \citep[cf.,][]{Uryasev00}. Here, we adopt the
sampling approach, which following \citet[Section 2.1]{Alexei} leads to the VaR portfolio
problem~\eqref{nominal} being written as:
\begin{equation}
\label{nominalsampled}
\begin{array}{llll}
z_{\VaR}: = &\min & -\nu \\[2ex]
&\st  & \nu = \min^{\lfloor \alpha m \rfloor+1}\{x{\tr}\xi^1,\dots,x{\tr}\xi^m\}\\
&     & x{\tr}\mu \ge \mu_0\\
&     & x{\tr}\one = 1\\
&     & x \in \mathcal{X} \subseteq \R^n_+, \nu \in \R,\\
     \end{array}
\end{equation}
where $\nu$ represents the $\VaR_{\alpha}(x\tr \xirnd)$,  the
vector of mean return estimates is, for simplicity, considered to be given by $\mu := (1/m)\sum_{j=1}^m \xi^j$. However, our results
are independent of this choice, and a variety of other estimation methods can be used
\citep[see, e.g.,][]{Black, Meuc07}. Also, for $k\in\{1,\dots,m\}$, and $u^j\in \R$, $j=1,\dots,m$, the
$k$-th smallest element in $\{u^1,\dots,u^m\}$ is denoted by $\min^{k}\{u^1,\dots,u^m\}$
(i.e., $\min^{k}\{u^1,\dots,u^m\}$ is the $k$-th order statistic $u^{(k)}$ in $\{u^1,\dots,u^m\}$).

Problem~\eqref{nominalsampled} is equivalent \citep[see, e.g.,][]{Benati07, FengWS15} to the
following mixed-integer linear program (MILP):
\begin{equation}
\label{eq:binaryprog}
\begin{array}{llll}
z_{\text{VaR}}= & \max & \nu \\[2ex]
&\st  & \displaystyle \sum_{j=1}^m y_j = \lfloor \alpha m \rfloor\\
&     &  My_j \ge \nu - x{\tr}\xi^j, & j=1,\dots,m\\
&     & x{\tr}\mu \ge \mu_0\\
&     & x{\tr}\one = 1\\
&     & x \in \mathcal{X} \subseteq \R^n_+, \nu \in \R\\
&     & y_j \in \{0,1\}, & j=1,\dots,m,\\
     \end{array}
     \end{equation}
where $M$ is a big enough constant (i.e., $M > 2\max\{|\xi^j_i|: i \in \{1,\dots,n\}, j \in \{1,\dots,m\}\}$), and as in~\eqref{nominalsampled},
$\nu$ represents the VaR of the portfolio~$x \in \R^n_+$. The extra binary variable $y_j$ denotes whether
$\nu$ is to the right ($y_j=1$) or to the left ($y_j=0$) of the sample portfolio
return $x{\tr} \xi^j$, for $j=1,\dots,m$.

In the literature, it is common to consider two alternate formulations of
the portfolio allocation problem. That is, besides the portfolio allocation formulation
in which one seeks to obtain the minimum risk portfolio satisfying a reward
benchmark (as in eq. \eqref{nominal} above), the alternate formulation in which
one seeks to obtain the maximum reward portfolio satisfying a risk benchmark
is commonly considered. It is well-known that these alternate formulations of the portfolio
problem are equivalent for the classical mean-variance portfolio model of \citet{Markowitz}
\citep[see, e.g.,][]{Krok02}. Due to the non-convexity of the VaR risk measure,
it is not surprising that this equivalence does not hold in general for the VaR portfolio problem
considered here. However, the relationship between these two alternate formulations
of the VaR portfolio problem is fundamental to develop the algorithm
presented here to address the solution of this problem. Below, we formally
present the alternate maximum reward portfolio satisfying a minimum VaR risk
benchmark $\tilde{\nu} \in \R$.

\begin{equation}
\label{nominalmax}
\begin{array}{lll}
\max & \mathbb{E}(x{\tr}\xirnd)  \\[2ex]
\st  & -\VaR_{\alpha}(x{\tr}\xirnd) \leq \tilde{\nu}\\
     & x{\tr}\one = 1\\
     & x \in \mathcal{X} \subseteq \R^n_+.\\
     \end{array}.
\end{equation}
Using the sampling approach, and similar to \eqref{nominal}, problem~\eqref{nominalmax} can be reformulated as the following MILP:
\begin{equation}
\label{maxsamp}
\begin{array}{llll}
z^*_{\text{VaR}} = & \max & x{\tr}\mu  \\[2ex]
&\st  & \displaystyle \sum_{j \in I} y_j \leq \lfloor \alpha m \rfloor\\
&     & My_j \ge \tilde{\nu} - x{\tr}\xi^j, & j=1,\ldots,m\\
&     & x{\tr}\one = 1\\
 &    & x \in \mathcal{X} \subseteq \R^n_+, v \in \R\\
 &    & y_j \in \{0,1\}, & j=1,\ldots,m.\\
\end{array}
\end{equation}

The relationship between the two alternative formulations
of the VaR portfolio problems~\eqref{eq:binaryprog} and~\eqref{maxsamp} will be analyzed in
the next section. Moreover, in Section~\ref{algorithm}, this relationship is exploited to obtain
approximate solutions of the VaR portfolio problem \eqref{eq:binaryprog} with a near-optimallity guarantee.

\section{On alternate portfolio allocation problem  formulations}
\label{S:Dual}
In portfolio allocation problems one seeks to find
the
portfolio with minimum risk subject to a constraint on the minimum level of
the portfolio's reward. Alternatively, the portfolio allocation problem is also formulated
as the problem of obtaining the portfolio with maximum reward subject to a constraint on the maximum level
of the portfolio's risk. Similar to \cite{Krok02}, these two problems
can be formally and respectively stated as follows:
\begin{equation}
\label{eq:riskopt}
\begin{array}{llll}
\beta(a) = & \min & \phi(x) \\
                               & \st     & R(x) \ge a\\
                               &          & x \in \X,\\
\end{array}
\end{equation}

\begin{equation}
\label{eq:retopt}
\begin{array}{llll}
\alpha(b) = & \max & R(x) \\
                               & \st     & \phi(x) \le b\\
                               &          & x \in \X,\\
\end{array}
\end{equation}
where $x \in \R^n$ represents the portfolio of assets, $\phi(x): \R^n \to \R$ measures the portfolio's risk,
$R(x): \R^n \to \R$ measures the portfolio's reward, and $\X \in \R^n$ represents the set of admissible
portfolios (e.g., $\X$ could contain long only positions constraints or benchmark constraints). Also,
$a, b, \in \R$, respectively represent the minimum required reward, and the maximum allowed risk. Throughout, we assume that the set $\X \in \R^n$ is compact (any position on an asset is typically constrained to be within certain lower and upper bounds), and
use the usual convention $\beta(a) = +\infty$ (resp. $\alpha(b) = -\infty$) if problem \eqref{eq:riskopt} (resp. \eqref{eq:retopt}) is
infeasible.

For the classical mean-variance portfolio allocation
model introduced by \citet{Markowitz}, it is well known that there is a one-to-one
correspondence between the optimal portfolios obtained from these two
models (i.e., \eqref{eq:riskopt} and \eqref{eq:retopt}). In more generality, it has been recently shown by \citet[][Thm. 3]{Krok02}
that this type of one-to-one relationship will hold  in more generality whenever the risk measure $\phi(x)$ is
convex and the reward measure $R(x)$ is concave.

Not surprisingly, when the risk measure $\phi(x)$ is defined by the portfolio's VaR, there is not a one-to-one
correspondence between the portfolio allocation models \eqref{eq:riskopt} and  \eqref{eq:retopt}.
However, as it will be illustrated therein, when using VaR as a risk measure, relaxations of both these problems
are useful in addressing the solution of \eqref{eq:riskopt}. Given this, and the fact that is has been erroneously
reported in \cite{Lin09} that there is a one-to-one correspondence between the portfolio allocation problems
\eqref{eq:binaryprog} and  \eqref{maxsamp}, it is worth to study below the relationship between these two
problems in a general setting when the risk measure $\phi(x)$ is not necessarily convex and/or the measure of reward $\R(x)$ is
not necessarily concave; extending \citet[][Thm. 3]{Krok02} to provide both sufficient and
necessary conditions for both~\eqref{eq:riskopt} and  \eqref{eq:retopt} to have a one-to-one correspondence.
These results are formally stated in the remainder of this section.

We define (recall that by assumption $\X \subseteq \R^n$ is compact) the minimum risk and maximum reward that any portfolio in the admissible set $\X \subseteq \R^n$ can have as:
\newcommand{\amax}{\bar a}

\begin{equation}
\begin{array}{llll}
\bmin = & \min & \phi(x) \\
                               & \st    & x \in \X\\
\end{array} \text { and }
\begin{array}{llll}
\amax = & \max & R(x) \\
                               & \st    & x \in \X.\\
\end{array}
\end{equation}
Theorem \ref{th:duality} below, provides sufficient and necessary conditions for a one-to-one correspondence between
the portfolio allocation problems \eqref{eq:riskopt} and  \eqref{eq:retopt}.

\begin{theorem}\label{th:duality} Let $I \subseteq [\alpha(\bmin),\amax]$ be an interval.
The relation  $a = \alpha(\beta(a))$ holds for any $a \in I$  if and only if $\beta(a)$ is
strictly increasing for all $a \in I$.
\end{theorem}

\begin{proof}
First notice that $\beta(a)$
is non-decreasing as a function of $a$.
Now,  assume
that there exists $a_1 \in I $ such that
$a_1 < \alpha(\beta(a_1))=:a_2$. Then
$\beta(a_1) \le \beta(a_2)$  as $\beta(\cdot)$ is non-decreasing, and
$\beta(a_2) = \beta(\alpha(\beta(a_1)))
\le \beta(a_1)$  as $\beta(\alpha(b)) \le b$ for all $b$.
Therefore, $\beta(a_1) = \beta(a_2)$.
To prove the other direction, assume $\beta(a)$ is not strictly increasing in $I$.
 Then, there exist $a_1,a_2 \in I$ with  $ a_1 < a_2 $ such that $\beta(a_1) = \beta(a_2)$. Then, using that $\alpha(\beta(a)) \ge a$ for all $a$, we obtain
$\alpha(\beta(a_1)) = \alpha(\beta(a_2)) \ge a_2 > a_1$.
(See Figure~\ref{fig:var_risk} for an illustration of the proof.)
\end{proof}

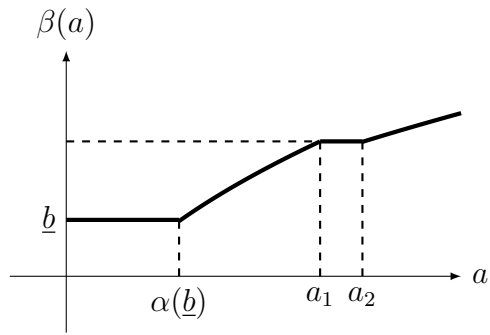
\begin{figure}[htb!]
\begin{center}
\begin{tikzpicture}[scale=0.75]
  \draw[smooth,samples=100,domain=0:1,  thick, color = black, dashed]
 plot(2,\x);
 \node[below] at (2,0){$\alpha(\bmin)$};
 \node[left] at (0,1){$\bmin$};
  \draw[smooth,samples=100,domain=0:2.3926,  thick, color = black, dashed]
 plot(4.5,\x);
  \draw[smooth,samples=100,domain=0:2.3926,  thick, color = black, dashed]
 plot(5.25,\x);
 \node[below] at (4.5,0){$a_1$};
\node[below] at (5.25,0){$a_2$};
 \draw[smooth,samples=100,domain=0:4.5,  thick, color = black, dashed]
 plot(\x,2.3926);
   \draw[smooth,samples=100,domain=0:2, ultra thick, color = black]
  plot(\x,1);
     \draw[smooth,samples=100,domain=2:4.5, ultra thick, color = black]
  plot(\x,2*\x^0.5-1.85);
     \draw[smooth,samples=100,domain=4.5:5.25, ultra thick, color = black]
  plot(\x, 2.3926);
     \draw[smooth,samples=100,domain=5.25:7, ultra thick, color = black]
  plot(\x,\x^0.6-0.32);
  \draw[-latex,color=black] (-1,0) -- (7,0);
  \draw[-latex,color=black] (0,-1) -- (0,4);
 \node[right] at (7,0){$a$};
 \node[above] at (0,4){$\beta(a)$};
\end{tikzpicture}
\end{center}
\caption{Illustration of Theorem~\ref{th:duality}.}
\label{fig:var_risk}
\end{figure}

As mentioned before, it is shown in \cite[][Thm. 3]{Krok02} that convexity in the risk measure, and concavity in the reward, provides sufficient conditions for Theorem~\ref{th:duality} to hold. This result can be obtained as a corollary of Theorem~\ref{th:duality}.

\newcommand{\amin}{\underline a}
\begin{corollary}
\label{cor:urasyev}
Let $\phi(x): \R^n \to \R$ be convex and $R(x): \R^n \to \R$ be concave. Assume $\{a \in [\alpha(\bmin),\amax]:\beta(a) > \bmin\}$ is non-empty and let $\amin = \inf\{a \in [\alpha(\bmin),\amax]:\beta(a) > \bmin\}$
Then $a = \alpha(\beta(a))$ for any $a \in [\amin,\amax] $.
\end{corollary}

\begin{proof}
From Theorem~\ref{th:duality} is enough to show that $\beta(\cdot)$ is strictly increasing on $(\amin,\amax]$.
For sake of contradiction,  let $\amin < a_1 < a_2 \le \amax$ be such that $\beta(a_1) = \beta(a_2)$. Let $x_i := \argmin\{\phi(x): R(x) \ge a_i, x \in \X\}$ for $i=1,2$. Thus
$\phi(x_1) = \phi(x_2)$. Let $\hat x$ be the optimal min-risk portfolio, i.e. $\phi(\hat x) = \bmin$ and $R(\hat x) = \alpha(\bmin)$. Let $\epsilon := \frac{a_2-a_1}{2a_2 - a_1 - \alpha(\bmin)}$. Then $0 < \epsilon < 1$. Let
$x' = \epsilon \hat x + \epsilon x_1 + (1-2\epsilon)x_2$. From the convexity of $\phi$, we get
that
\[\phi(x') = \phi(\epsilon \hat x + \epsilon x_1 + (1-2\epsilon)x_2) \le \epsilon \phi(\hat x) + \epsilon \phi(x_1) + (1-2\epsilon)\phi(x_2) = \epsilon \bmin + (1-\epsilon)\phi(x_1)< \phi(x_1).\]
Also, by the concavity of $R(x)$, we get that
\[R(x') \ge \epsilon R(\hat x) + \epsilon R(x_1) + (1-2\epsilon)R(x_2) \ge \epsilon \alpha(\bmin) + \epsilon a_1 + (1-2\epsilon) a_2 = a_1.\]
Thus, $x_1 \not = \argmin\{\phi(x): R(x) \ge a_1, x \in \X\}$,
a contradiction.
\end{proof}

\citet[][Thm. 3]{Krok02}
assume a regularity condition for each value of the pair~$(a,b)$.
In contrast, in
Corollary~\ref{cor:urasyev}, the ranges of $a$ and $b$
for which the one-to-one correspondence between the
alternative formulations holds is precisely characterized.

Although sufficient, the convexity condition in  Corollary~\ref{cor:urasyev} is not necessary to have the one-to-one correspondence
between the portfolio allocation problems \eqref{eq:riskopt} and  \eqref{eq:retopt}. To illustrate this, we consider
the following simple example in which the risk measure $\phi(x)$ is related to the well-known Huber's function \cite[see, e.g.,][]{ROStats} that commonly appears in robust statistics.

\begin{example}
\label{ex:simple} Let $\kappa > 1$ be given.
Let the functions $\phi: \R \to \R$ and $R: \R \to \R$ be given by
\[
\phi(x) = \left \{     \begin{array}{ll}  x^2 & \text{if $|x| \le \kappa$}\\
                                                      x + \kappa(\kappa-1) &  \text{if $|x| \ge \kappa$}\\ \end{array} \right .,
\]
and $R(x) = x$. Also, let
the set $\mathcal{X} = [-2\kappa, 2\kappa]$. The function $\phi(x)$ is not convex, as $2\phi(\kappa) = 2\kappa^2 > (\kappa -1)^2 + \kappa + 1 + \kappa(\kappa -1) = \phi(k-1) + \phi(\kappa +1)$ (see Figure \ref{fig:exsimple} (left)). Thus the conditions of Corollary~\ref{cor:urasyev} are not satisfied. However, it is easy to see that the
function $\beta(a)$ is strictly increasing in the domain $a \ge \alpha(\bmin) = 0$ (see Figure \ref{fig:exsimple} (right)). Note that by changing the domain $\mathcal{X} = \R_+$ and $R(x) = x^2$ one has a similar example where $\beta(a)$ is strictly increasing but now neither $\phi(x)$ is convex nor $R(x)$ is concave.
\end{example}

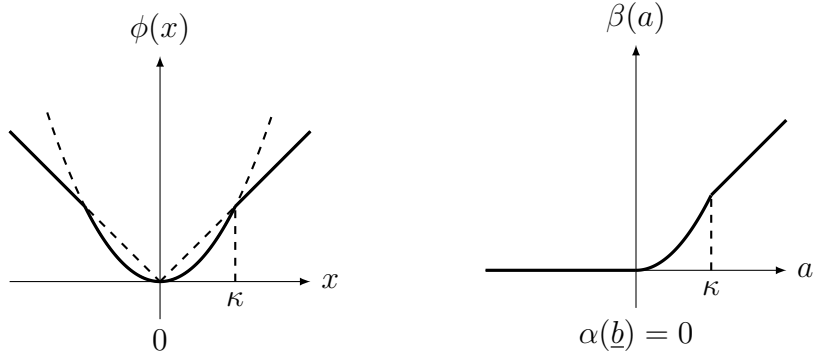
\begin{figure}[htb!]
\begin{center}
\begin{tikzpicture}
  \draw[smooth,samples=100,domain=-1.5:1.5,  thick, color = black, dashed]
 plot(\x,\x*\x);
   \draw[smooth,samples=100,domain=0:2,  thick, color = black, dashed]
 plot(\x,\x);
  \draw[smooth,samples=100,domain=-2:0,  thick, color = black, dashed]
 plot(\x,-\x);
  \draw[smooth,samples=100,domain=0:1,  thick, color = black, dashed]
 plot(1,\x);
  \draw[smooth,samples=100,domain=1:2,  thick, color = black, very thick]
 plot(\x,\x);
  \draw[smooth,samples=100,domain=-2:-1,  thick, color = black, very thick]
 plot(\x,-\x);
  \draw[smooth,samples=100,domain=-1:1,  thick, color = black, very thick]
 plot(\x,\x*\x);
  \draw[-latex,color=black] (-2,0) -- (2,0);
  \draw[-latex,color=black] (0,-0.5) -- (0,3);
 \node[right] at (2,0){$x$};
 \node[above] at (0,3){$\phi(x)$};
 \node[below] at (1,0){$\kappa$};
 \node[below] at (0,-0.5){$0$};
\end{tikzpicture} \qquad \qquad
\begin{tikzpicture}
  \draw[smooth,samples=100,domain=1:2,  thick, color = black, very thick]
 plot(\x,\x);
  \draw[smooth,samples=100,domain=-2:0,  thick, color = black, very thick]
 plot(\x,0);
  \draw[smooth,samples=100,domain=0:1,  thick, color = black, very thick]
 plot(\x,\x*\x);
  \draw[smooth,samples=100,domain=0:1,  thick, color = black, dashed]
 plot(1,\x);
  \draw[-latex,color=black] (-2,0) -- (2,0);
  \draw[-latex,color=black] (0,-0.5) -- (0,3);
 \node[right] at (2,0){$a$};
 \node[above] at (0,3){$\beta(a)$};
  \node[below] at (1,0){$\kappa$};
  \node[below] at (0,-0.5){$\alpha(\bmin) = 0$};
\end{tikzpicture}
\end{center}
\caption{Illustration of Example~\ref{ex:simple}. Function $\phi(x)$ (left), and
corresponding $\beta(a)$ (right).}
\label{fig:exsimple}
\end{figure}

As mentioned earlier, when the risk measure $\phi(x)$ is defined by the
portfolio's return~VaR, there is in general not a one-to-one correspondence between
the portfolio allocation problems  \eqref{eq:riskopt} and  \eqref{eq:retopt}. This is
formally stated in the next remark, which corrects the erroneous characterization
between problems \eqref{eq:riskopt} and  \eqref{eq:retopt} given in \cite{Lin09}.

\begin{remark}
When the risk measure $\phi(x)$ in \eqref{eq:riskopt} is defined by the
portfolio's return Value-at-Risk (VaR)
$\beta(a)$ is
not in general strictly increasing (in the domain $a \ge \alpha(\bmin)$). This is illustrated with
the numerical example below.
\end{remark}

\begin{example}
Consider the instance of problem  \eqref{eq:riskopt} in which $x \in \R^2$ represents
the percentage of money invested in the two assets Microsoft {\tt (MSFT)} and 3M {\tt (MMM)}.
Let $\X = \{x \in \R^2_+: x_1 + x_2 = 1\}$. Also, let $\phi(x)$ and $R(x)$ respectively be
the estimates of the portfolio's return VaR$_{5\%}$ and expected portfolio return based on
a sample monthly returns of  {\tt (MSFT)} and {\tt (MMM)},
from April 1986 to December 2006 (source Wharton Research Data Services (WRDS)).
After computing $\beta(a)$ in \eqref{eq:riskopt} one obtains Figure~\ref{fig:numVaR}.

\begin{figure}[!tbh]
\begin{center}
\includegraphics[width = 3in]{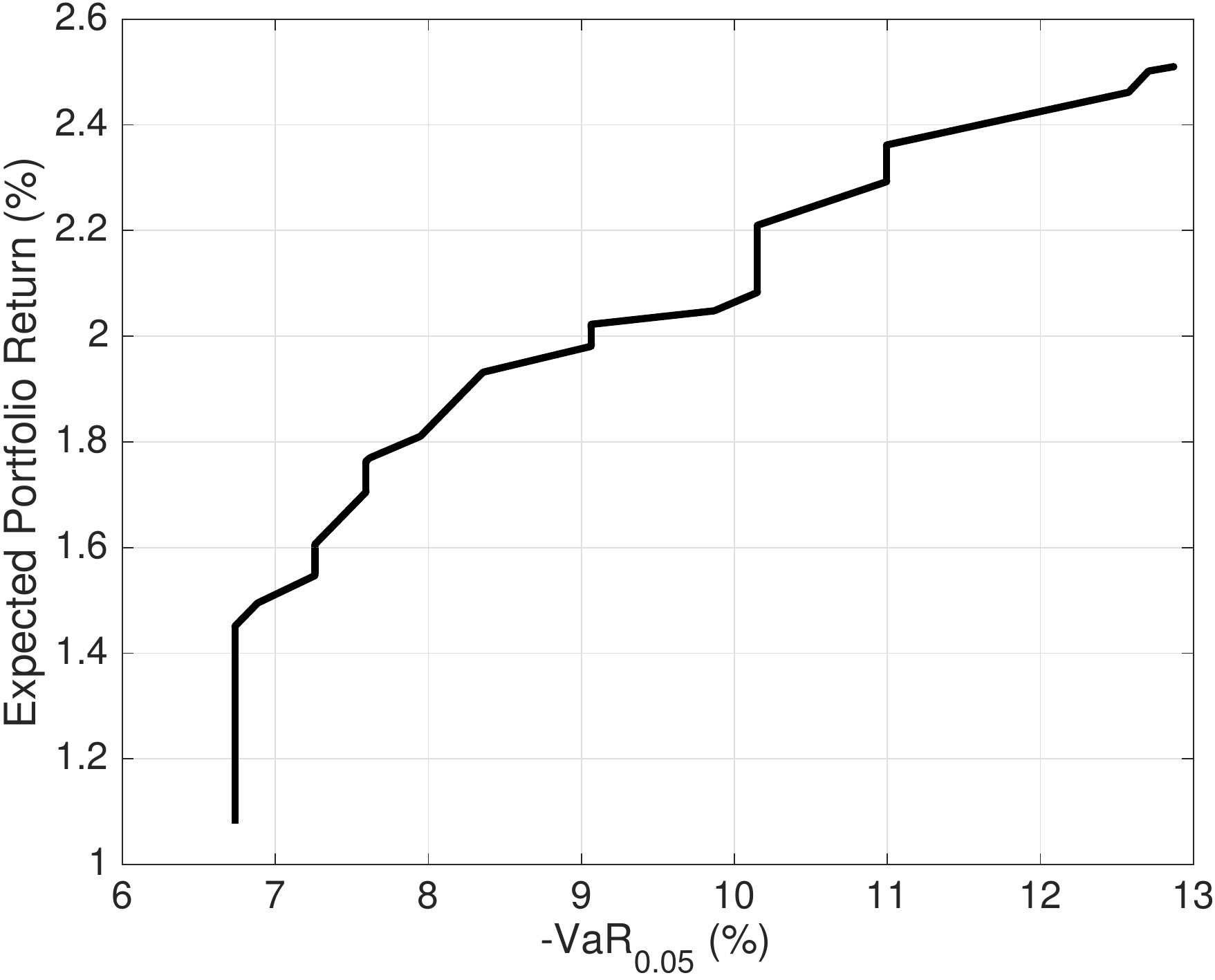}
\end{center}
\caption{Instance of $\beta(a)$ (cf., \eqref{eq:riskopt}) not being strictly increasing when the portfolio's risk measure is the
VaR of the portfolio returns.}
\label{fig:numVaR}
\end{figure}

Note that the areas of Figure~\ref{fig:numVaR} in which the risk remains constant while the expected portfolio return increases show that the VaR is not strictly increasing as a function of the expected portfolio return.

\end{example}

We finish this section by showing that one can take advantage of the alternative formulations \eqref{eq:riskopt}
and \eqref{eq:retopt} to obtain a measure of the closeness to optimality of a feasible solution of \eqref{eq:riskopt}
when an appropriate bound on the optimal value of \eqref{eq:retopt} can be obtained.

\begin{lemma}
\label{lem:bound}
Let $a \le \amax$. If $\alpha(b) < a$, then $ \beta(a) \ge b$.
\end{lemma}

\begin{proof}
Notice that if $b < \bmin$ then by definition $\beta(a) \ge \bmin > b$. Thus we assume $b \ge \bmin$. For the sake of contradiction
assume $\beta(a) <b$. Then there exists $x \in \X$ such that $R(x) \ge a$ and
$\phi(x) < b$. Thus $x$ is a feasible solution for \eqref{eq:retopt}, which implies $\alpha(b) \ge a$, a contradiction.
\end{proof}

\begin{proposition}
\label{p:bound}
Given $a \le \amax$. Let $\tilde{x} \in \R^n$ be a feasible solution of \eqref{eq:riskopt} and $\delta \ge 0$ be a given tolerance. If $\alpha(\phi(\tilde{x})-\delta) < a$. Then $\phi(\tilde{x})-\delta \le \beta(a) \le \phi(\tilde{x})$.
\end{proposition}

In what follows, we use Proposition~\ref{p:bound} to provide an algorithm to address the
solution of the VaR portfolio problem. As mentioned earlier, in this case, solving the associated minimum risk portfolio problem~\eqref{eq:riskopt} or the maximum return portfolio problem \eqref{eq:retopt}
to optimality is inherently difficult.

\section{The algorithm}
\label{algorithm}

In this section, we provide an algorithm to obtain approximate solutions for the MILP formulation
of the VaR portfolio problem \eqref{eq:binaryprog}. First, the goal of the algorithm is to find a
near-optimal feasible solution for~\eqref{eq:binaryprog} (cf., Section~\ref{sec:lower}).
Next, the goal is to provide a near-optimality guarantee for this feasible solution (cf., Section~\ref{sec:upper}).

\subsection{Lower bound for optimal VaR}
\label{sec:lower}
Let us denote $[m] := \{1,\dots,m\}$. Now, given $J \subseteq [m]$,  let $J^c := [m] \setminus J$,
and consider the following problem:
\begin{equation}
\label{eq:lowJ}
\begin{array}{llll}
\underline{z}_J:= & \max & \nu \\[2ex]
&\st  & \displaystyle \sum_{j \in J} y_j = \lfloor \alpha m \rfloor\\
&     & My_j \ge \nu - x{\tr}\xi^j, & j \in J\\
&     & 0 \ge \nu- x{\tr}\xi^j, & j \in J^c\\
&     & x{\tr}\mu \ge \mu_0\\
&     & x{\tr}\one = 1\\
&     & x \in \mathcal{X} \subseteq \R^n_+, v \in \R\\
&     & y_j \in \{0,1\}, & j \in J.\\
\end{array}
\end{equation}
Note that $\eqref{eq:lowJ}$ is the optimization problem obtained from \eqref{eq:binaryprog} by setting $y_j = 0$ for all~$j \in J^c$. Hence $\underline{z}_J \leq z_{\text{VaR}}$ for all $J \subseteq [m]$. In Algorithm~\ref{alg:Algorithm-1} below, the formulation~$\eqref{eq:lowJ}$, together with an appropriate update of the set $J$, is used iteratively to obtain near-optimal feasible solutions to~\eqref{eq:binaryprog}. Specifically, after setting an initial set $J = J_0 \subset [m]$ problem~\eqref{eq:lowJ} is solved. Let~$y^J \in \{0,1\}^{|J|}$ be the optimal value of the binary variables of  \eqref{eq:lowJ}. These values are used to construct the linear program below
obtained by fixing the binary variables $y \in \{0,1\}^m$ in~\eqref{eq:binaryprog} such
that $y_J = y^J$, and $y_j = 0$, for all $j \in J^c$.
\begin{equation}
\label{eq:LP}
\begin{array}{llllll}
& \max & \nu \\[2ex]
&   \st  &  My^J_j \ge \nu - x{\tr}\xi^j, & j=1,\dots,m\\
&     & x{\tr}\mu \ge \mu_0\\
&     & x{\tr}\one = 1\\
&     & x \in \mathcal{X} \subseteq \R^n_+, \nu \in \R.\\
     \end{array}
     \end{equation}
After solving~\eqref{eq:LP}, the shadow prices associated with its big-M constraints (the first set of constraints in ~\eqref{eq:LP}) are used to update the set $J \subseteq [m]$. That is, the set $J$ is augmented by the samples' indices whose corresponding
big-M constraints in \eqref{eq:LP} have a positive shadow price. This type of update is similar to the one used when solving MILPs using {\em branch and price} techniques \cite[see, e.g.,][]{Trick07}.
As described in Algorithm~\ref{alg:Algorithm-1}, this procedure is applied iteratively until no further improvement in the lower bound of \eqref{eq:binaryprog} can be obtained. The VaR of the portfolio obtained at the end of the algorithm serves as a lower bound for the optimal VaR portfolio problem.

\begin{algorithm}[!htbp]
\caption{Lower bound of optimal VaR.}
\begin{algorithmic}[1]
\Procedure{Lower bound}{${\alpha, \mu_0,J_0 \subseteq [m], |J_0|  \geq \lfloor \alpha m \rfloor}$}
\State  $J \gets J_0$
\State $J^{\old} \gets [m]$
\While{$J^{\old} \cap J \neq J^{\old}$}
\State  $y^J \gets \arg_y(\underline{P}_J)$  \label{Linit}
\State $d \gets$ shadow prices of the big-M constraints in $\eqref{eq:lowJ}$
\State $J^{\old} \gets J$
\State  $J \gets \{i \in J: y^J_i =1\} \cup \{i \in J^c: d_i > 0\}$
\EndWhile
\State \textbf{return} $\tilde{x} \gets \arg_{x}\eqref{eq:lowJ}$  \Comment{feasible portfolio for~\eqref{eq:binaryprog} }
\State \textbf{return}  $\tilde{\nu} \gets \arg_{v}\eqref{eq:lowJ}$ \Comment{lower bound for~\eqref{eq:binaryprog} }
\State \textbf{return}  $\tilde{y} \gets \arg_{y}\eqref{eq:lowJ}$, $I_0 \gets \{i \in [m]: \tilde{y}_i = 1\}$ \Comment{to initialize Algorithm~\ref{alg:Algorithm-2} in Section~\ref{sec:upper}}
\EndProcedure
\end{algorithmic}
\label{alg:Algorithm-1}
\end{algorithm}

As it will be shown in Section~\ref{numerical}, Algorithm~\ref{alg:Algorithm-1} provides
a tighter lower bound~$\tilde{\nu} =\underline{z}_J$,  for the VaR portfolio problem~\eqref{eq:binaryprog} than those obtained using the CVaR-based algorithms proposed by~\cite{LarsMU02} in a comparable running time. More importantly, Algorithm~\ref{alg:Algorithm-1} provides a feasible solution $\tilde{x}, \tilde{\nu}$, for the VaR portfolio problem \eqref{eq:binaryprog} whose near-optimality can be guaranteed using Algorithm~\ref{alg:Algorithm-2}, which is presented in the next section.

\subsection{Upper bound for optimal return}
\label{sec:upper}
In this section, the aim is to obtain a measure of the closeness to optimality of the feasible solution $\tilde{x}, \tilde{v}$, for the VaR portfolio problem obtained by Algorithm~\ref{alg:Algorithm-1}.
For this purpose, we first apply Proposition~\ref{p:bound}  to the alternative formulations of the VaR portfolio problem~\eqref{eq:binaryprog} and~\eqref{maxsamp}.
Specifically, let $\delta >0$ be a specified tolerance, and  $\tilde{x} \in \R^n_+$ be a feasible portfolio for~\eqref{eq:binaryprog}, with associated VaR $\tilde{\nu}$; that is,
$\tilde{\nu} = {\min}^{\lfloor \alpha m \rfloor+1}\{\tilde{x}{\tr}\xi^1,\dots,\tilde{x}{\tr}\xi^m\}$.
Then, from Proposition~\ref{p:bound}, it follows that if the optimal value of~\eqref{maxsamp} satisfies
\begin{equation}
\label{eq:condition}
z^*_{\text{VaR}} < \mu_0 \Rightarrow
\tilde{\nu} - \delta \le z_{\text{VaR}}  \le \tilde{\nu}.
\end{equation}
That is, the near-optimality of the feasible portfolio $\tilde{x} \in \R^n_+$ to the optimal portfolio corresponding to the VaR portfolio problem~\eqref{eq:binaryprog}, can be measured in terms of $\delta\in \R_{++}$.

Clearly, directly solving~\eqref{maxsamp} to check whether condition $z^*_{\text{VaR}} < \mu_0$ in \eqref{eq:condition}
holds for a feasible portfolio $\tilde{x} \in \R^n_+$ of~\eqref{eq:binaryprog} is as computationally inefficient as directly solving~\eqref{eq:binaryprog}. Therefore, we present an appropriate upper bound for the alternative formulation of the VaR portfolio problem~\eqref{maxsamp} that allows to efficiently guarantee the near-optimality of the feasible solutions of the VaR problem obtained after using Algorithm~\ref{alg:Algorithm-1}. Specifically,
given $I \subseteq [m]$ with $|I| \ge \lfloor \alpha m \rfloor$ and $\tilde{\nu}$, a lower bound~\eqref{eq:binaryprog} (i.e, $\tilde{\nu} \le z_{\text{VaR}}$),
 consider the problem
\begin{equation}
\label{uppI}
\begin{array}{llll}
\overline{\mu}_I:= & \max & x{\tr}\mu  \\[2ex]
&\st  & \displaystyle \sum_{j \in I} y_j \leq \lfloor \alpha m \rfloor\\
&     & My_j \ge \tilde{\nu} - x{\tr}\xi^j, & j \in I\\
&     & x{\tr}\one = 1\\
&     & x \in \mathcal{X} \subseteq \R^n_+, \\
&     & y_j \in \{0,1\}, & j \in I.\\
\end{array}
\end{equation}
Notice that
$\overline{\mu}_{[m]} = z^*_{\text{VaR}}$ (cf., \eqref{maxsamp}). Next, we
show that~\eqref{uppI} is a relaxation of~\eqref{maxsamp}.

\begin{proposition}
\label{prop:dual}
Let $I \subseteq [m]$ with $|I| \ge \lfloor \alpha m \rfloor$. Then problem~\eqref{uppI}
is a relaxation of~\eqref{maxsamp}. That is, $\overline{\mu}_I \ge z^*_{\VaR}$.
\end{proposition}
\begin{proof}
Let $x \in \R^n_+, y \in \{0,1\}^m$ be feasible for~\eqref{maxsamp}, then we have that
$x\tr\one =1$, and $x\in \mathcal{X}$. Moreover, the fact that there exist
$y \in \{0,1\}^m$ such that $\sum_{j \in [m]} y_j \leq \lfloor \alpha m \rfloor$,
and $My_j \ge \tilde{\nu} - x{\tr}\xi^j,$ for all~$j \in [m]$, implies that
$\tilde{\nu} \le  {\min}^{\lfloor \alpha m \rfloor+1}\{x{\tr}\xi^j: j \in [m]\}
\le {\min}^{\lfloor \alpha m \rfloor+1}\{x{\tr}\xi^j: j \in I\}$. Thus, $y_I \in \{0,1\}^{|I|}$
satisfies $\sum_{j \in I} y_j \leq \lfloor \alpha m \rfloor$,
and $My_j \ge \tilde{\nu} - x{\tr}\xi^j,$ for all $j \in I$. Thus,~$(x, y_I)$ is a feasible
solution for~\eqref{uppI} with objective value $x\tr \mu$.
\end{proof}

Notice that from Proposition~\ref{prop:dual}, it follows that
\[\overline{\mu}_I < \mu_0
\Rightarrow z^*_{\VaR} < \mu_0.\]
In Algorithm~\ref{alg:Algorithm-2} below, we take advantage of this fact by iteratively using the formulation $\eqref{uppI} $, together with an appropriate update of the set $I$,
with the aim of showing the near-optimality of the feasible solution
of the VaR portfolio problem~\eqref{eq:binaryprog} obtained from Algorithm~\ref{alg:Algorithm-1}. The set $I$ is updated by heuristically
adding samples from the set $[m] \setminus I$ (see, Algorithm~\ref{alg:Algorithm-2} for details) until condition~\eqref{eq:condition} is satisfied.

{\tiny
\begin{algorithm}[!htb]
\caption{Upper bound for optimal return}
\begin{algorithmic}[1]
\Procedure{Upper bound}{${\alpha, \beta, \mu_{0}, \delta, \text{Iter}^{\max}, \text{ and } \tilde{x}, \tilde{\nu}, I_0 \text{ from Algorithm~\eqref{alg:Algorithm-1}}}$}
\State $m' \gets \lfloor (\beta \alpha m) \rfloor$
\State $I \gets I_0$,
\State $\nu \gets \tilde{\nu}+\delta$
\State $\overline{\mu}_I \gets \tilde{x}\tr \mu$
\While{$\overline{\mu}_I  \geq \mu_0, I \subset [m], \text{ and Iter}\le  \text{Iter}^{\max}$}
\State  $\overline{\mu}_I \gets$ objective value of~\eqref{uppI} \Comment{$+\infty$ if~\eqref{uppI} is infeasible}
\State $x \gets \arg_{x}\eqref{uppI}$ \Comment{optimal portfolio for~\eqref{uppI}}
\State $I^{\rm{old}} \gets I$
\State $\nu' \gets \min^{m'+1}\{x\tr \xi^j: j \in [m]\setminus I_0\}$
\State $I \gets I^{\rm{old}} \cup \{j \in [m] \setminus I^{\rm{old}}: x\tr \xi^j \le \nu' \}$
\EndWhile

\If{$\text{Iter} <  \text{Iter}^{\max}$}
\State {\bf return} The $\delta$ near-optimality of $\tilde{x}, \tilde{\nu}$ is proven
\Else
\State {\bf return} The $\delta$ near-optimality of $\tilde{x}, \tilde{\nu}$ could not be verified
\EndIf
\EndProcedure
\end{algorithmic}
\label{alg:Algorithm-2}
\end{algorithm}
}

\section{Numerical Results}
\label{numerical}
In this section, we present numerical results to compare the performace of Algorithm~\ref{alg:Algorithm-1} against the CVaR-based algorithms proposed by~\cite{LarsMU02} to obtain lower bounds on the VaR portfolio problem~\eqref{eq:binaryprog}. Moreover, we compare the performance of Algorithm~\ref{alg:Algorithm-1} and
Algorithm~\ref{alg:Algorithm-2} to obtain guaranteed near-optimal solutions
for the VaR portfolio problem~\eqref{eq:binaryprog}, against directly solving~\eqref{eq:binaryprog} using state-of-the-art mixed integer linear programming (MILP) solvers.

To carry out the experiments we use the daily returns data from
Kenneth R. French's website {\small \url{http://mba.tuck.dartmouth.edu/pages/faculty/ken.french}} on 100 portfolios formed on size and book-to-market (10 x 10). The data can be downloaded at {\small \url{http://mba.tuck.dartmouth.edu/pages/faculty/ken.french/ftp/100_Portfolios_10x10_Daily_TXT.zip}}. From this data, instances of the VaR portfolio problem~\eqref{eq:binaryprog} having 
number
of assets $n \in [30,90]$, number of samples $m \in [1000, 3500]$ (for every value of $n$), and
expected profit $\mu_0 \in \{\mu^{-} +\frac{i}{k+1}(\mu^{+}- \mu^{-}): i=1,\dots,k\}$, with $k=6$ and $\mu^+$ (resp. $\mu^-$) is the largest (lowest) asset mean return. Similar to~\cite{Benati07, FengWS15}, the parameter $\alpha \in (0,1)$ (cf., \eqref{vardef}) is set to the popular value used in practice of $\alpha = 0.01$. 

All the code necessary to create the instances of the optimization problems discussed throughout the article is implemented using {\tt Matlab 2016a} and the modelling language {\tt YALMIP}, which is available at
\url{users.isy.liu.se/johanl/yalmip/}. {\tt Gurobi 6.5.0}  is used to obtain the solution of all the necessary linear programs and MILPs on a Intel(R) Core (TM) i3-2310M CPU @ 2.10 GHz, 4GB RAM machine.

\subsection{Lower bound for portfolio's VaR}

We compare the performace of Algorithm~\ref{alg:Algorithm-1} against the CVaR-based algorithms proposed by~\cite{LarsMU02} to obtain lower bounds on the VaR portfolio problem~\eqref{eq:binaryprog}.

In all instances, Algorithm~\ref{alg:Algorithm-1} is initialized by setting
$J_0$ as the first $m_0 = \lceil 2\alpha m \rceil$ samples of the instance.
Also, the algorithms being compared are allowed to run for a maximum time of up 3600 seconds.

The lower bound results on the VaR portfolio problem~\eqref{eq:binaryprog} obtained by the three~(3) algorithms
are summarized in Table~\ref{tab:lowbounds}, Figure~\ref{fig_LB_1}, \ref{fig_LB_2}, and~\ref{fig_LB_3}. In Table~\ref{tab:lowbounds},
for every combination of the number of assets ($n$) and the number of samples~($m$), an average is taken
over the instances with different values of $\mu_0$, between the values  $\mu_0^{\min}$ and  $\mu_0^{\max}$.
For each algorithm, the column ``gap'', indicates the relative percentage
error between the lower bound on the VaR portfolio problem~\eqref{eq:binaryprog} and its optimal solution provided by solving the MILP~\eqref{eq:binaryprog}. In the few instances when the
MILP~\eqref{eq:binaryprog} cannot be solved within the maximum allowed time (of 3600 s.), the optimal solution of~\eqref{eq:binaryprog} is replaced by the best (higher) lower bound obtained from the lower bound algorithms. Thus, the gap columns in Table~\ref{tab:lowbounds} cleary show that
Algorithm~\ref{alg:Algorithm-1} provides tighter lower bounds on the VaR portfolio problem~\eqref{eq:binaryprog},
than the ones provided by Algorithm~1 and Algorithm~2~\cite[cf.,][]{LarsMU02}. Also, it is clear that the percentage by which
Algorithm~\ref{alg:Algorithm-1} provides tighter bounds than Algorithms~1 and~2 is substantial and ranges between $1\%-7\%$ on average. A more granular evidence of this result is shown in Figures~\ref{fig_LB_1} and~\ref{fig_LB_2}. In these figures, for each algorithm, the relative gap with respect to the optimal value of the VaR portfolio problem~\eqref{eq:binaryprog} for each of the instances considered is plotted in the $y$-axis, while
the $x$-axis labels indicate the values of the number of samples ($m$), and expected return ($\mu_0$) of the instance. Also, the number of assets ($n$) is indicated in each of the plots.  In the next section, the tightness of the bounds provided by Algorithm~\ref{alg:Algorithm-1} will be key to be able to guarantee the near-optimality of the solutions for the VaR portfolio problem~\eqref{eq:binaryprog} provided by Algorithm~\ref{alg:Algorithm-1}.

As shown in Table~\ref{tab:lowbounds}, the tighter bounds obtained by Algorithm~\ref{alg:Algorithm-1}, in comparison with Algorithm~1 and Algorithm~2 in~\cite{LarsMU02}, are obtained in comparable running times. As mentioned earlier, in
Table~\ref{tab:lowbounds},
for every combination of the number of assets ($n$) and the number of samples~($m$), an average is taken
over the instances with different values of $\mu_0$, between the values  $\mu_0^{\min}$ and  $\mu_0^{\max}$.
For each algorithm, the column ``$T/T^*$'', indicates the average (over instances with different values of $\mu_o$, and equal $n$, $m$) of the ratio between the time taken by each of the algorithms and the minimum of these times on an instance with a particular $\mu_0 \in [\mu_0^{\max}, \mu_0^{\min}]$.
 From these results it is clear that the average times of the three algorithms are mostly comparable. In Figure~\ref{fig_LB_3}, the average running time information of the algorithms is provided. It is clear from this figure that most of the time, the average time taken by the three algorithms is similar, and that even when there are significant differences between the times, such differences are not of significant practical relevance, since the times required by the algorithms are in the range of at most hundreds of seconds.

\begin{table}[htb]
\setlength{\abovecaptionskip}{10pt}
{\footnotesize
\begin{center}
\begin{tabular}{ccrccccccccccc}
\toprule
\multicolumn{2}{c}{Size} && \multicolumn{2}{c}{$\mu_0$}  && \multicolumn{2}{c}{Alg.~\ref{alg:Algorithm-1}} &&  \multicolumn{2}{c}{Alg.~1} &&  \multicolumn{2}{c}{Alg.~2}\\
\cmidrule{1-2} \cmidrule{4-5} \cmidrule{7-8} \cmidrule{10-11} \cmidrule{13-14}
$n$ & $m$ && $\min$ & $\max$ && \multicolumn{1}{c}{gap} & \multicolumn{1}{c}{$T/T^*$}&& \multicolumn{1}{c}{gap} & \multicolumn{1}{c}{$T/T^*$} &&\multicolumn{1}{c}{gap} & \multicolumn{1}{c}{$T/T^*$}\\
\midrule
30	&	1000	&&	0.019	&	0.058	&&	0.04	&	1.1	&&	1.23	&	1.4	&&	0.94	&	1.2	\\
30	&	1500	&&	0.043	&	0.071	&&	0.08	&	1.4	&&	1.37	&	1.2	&&	1.48	&	1.0	\\
30	&	2000	&&	-0.018	&	0.056	&&	0.29	&	1.1	&&	3.02	&	1.1	&&	2.32	&	1.0	\\
30	&	2500	&&	0.012	&	0.046	&&	0.00	&	1.1	&&	2.85	&	1.2	&&	2.56	&	1.0	\\
50	&	1000	&&	0.015	&	0.069	&&	0.00	&	1.0	&&	2.23	&	1.2	&&	1.13	&	1.1	\\
50	&	1500	&&	0.062	&	0.075	&&	0.14	&	1.3	&&	1.59	&	1.1	&&	1.59	&	1.0	\\
50	&	2000	&&	-0.018	&	0.056	&&	0.29	&	1.1	&&	2.87	&	1.1	&&	2.45	&	1.0	\\
50	&	2500	&&	0.012	&	0.068	&&	0.00	&	1.1	&&	3.04	&	1.0	&&	3.54	&	1.0	\\
50	&	3000	&&	0.005	&	0.054	&&	0.00	&	1.4	&&	7.00	&	1.0	&&	6.22	&	1.1	\\
50	&	3500	&&	0.012	&	0.058	&&	0.00	&	1.1	&&	1.14	&	1.0	&&	1.14	&	1.0	\\
70	&	1000	&&	0.015	&	0.069	&&	0.00	&	1.0	&&	1.50	&	2.6	&&	2.71	&	4.3	\\
70	&	1500	&&	0.039	&	0.075	&&	0.06	&	1.3	&&	0.84	&	1.5	&&	1.30	&	1.2	\\
70	&	2000	&&	-0.017	&	0.061	&&	0.00	&	1.2	&&	1.96	&	1.9	&&	1.80	&	1.8	\\
70	&	2500	&&	0.012	&	0.068	&&	0.00	&	2.1	&&	5.82	&	2.6	&&	4.78	&	2.4	\\
70	&	3000	&&	0.005	&	0.018	&&	0.00	&	6.0	&&	4.33	&	9.3	&&	2.32	&	1.0	\\
70	&	3500	&&	0.012	&	0.069	&&	0.12	&	2.1	&&	2.94	&	1.0	&&	2.73	&	1.1	\\
90	&	1000	&&	0.007	&	0.048	&&	0.00	&	1.1	&&	7.55	&	1.1	&&	2.37	&	1.1	\\
90	&	1500	&&	0.063	&	0.084	&&	0.00	&	1.5	&&	0.94	&	1.0	&&	0.94	&	1.0	\\
90	&	2000	&&	0.061	&	0.061	&&	0.02	&	1.5	&&	1.68	&	1.0	&&	0.28	&	1.0	\\
90	&	2500	&&	0.012	&	0.068	&&	0.28	&	3.4	&&	7.49	&	1.1	&&	6.28	&	1.0	\\
90	&	3000	&&	0.005	&	0.068	&&	0.00	&	2.9	&&	2.38	&	1.5	&&	2.09	&	1.0	\\
90	&	3500	&&	0.012	&	0.069	&&	0.00	&	3.5	&&	3.12	&	1.0	&&	3.43	&	1.1	\\\bottomrule\end{tabular}
\end{center}
}
\caption{Performance of Algorithm~\ref{alg:Algorithm-1} vs. Algorithm~1 and Algorithm~2 in~\citep{LarsMU02} to compute lower bounds on the VaR portfolio allocation problem~\eqref{eq:binaryprog}.
The column gap indicates the VaR lower bound \%~gap to the optimal VaR.
The column $T/T^*$, is the ratio between the time required to obtain the lower bound $T$ against the minimum time needed by the three algorithms $T^*$.
}
\label{tab:lowbounds}
\end{table}

\begin{figure}[!htb]
  \begin{subfigure}[b]{0.5\linewidth}
    \centering
    \includestandalone[width=\linewidth, height = \linewidth]{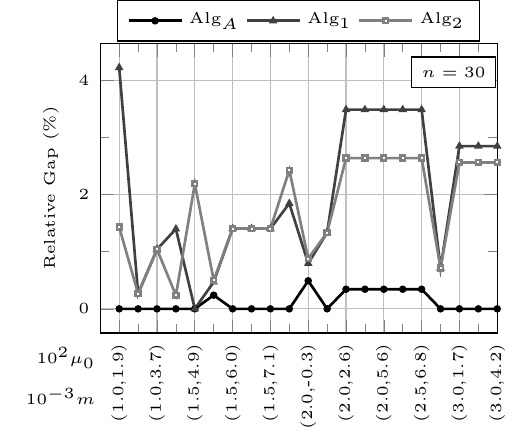}
  \end{subfigure}
  \begin{subfigure}[b]{0.5\linewidth}
    \centering
    \includestandalone[width=\linewidth,  height = \linewidth]{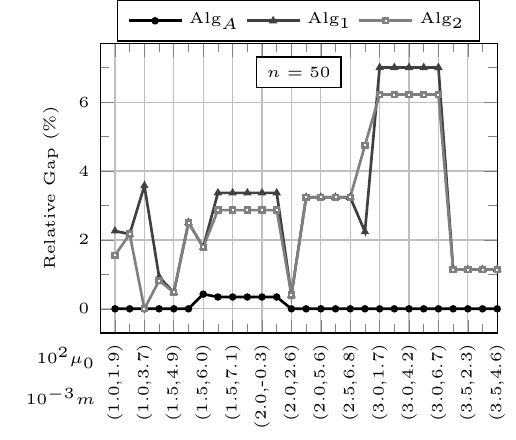}
  \end{subfigure}
    \caption{Comparison of the relative gap between the optimal value
     of the VaR portfolio problem~\eqref{eq:binaryprog} and the lower bounds for~\eqref{eq:binaryprog} provided by Algorithm~\ref{alg:Algorithm-1}, and Algorithms~1 and~2  by~\cite{LarsMU02}.}
  \label{fig_LB_1}
\end{figure}

\begin{figure}[!htb]
  \begin{subfigure}[b]{0.5\linewidth}
    \centering
    \includestandalone[width=\linewidth, height = \linewidth]{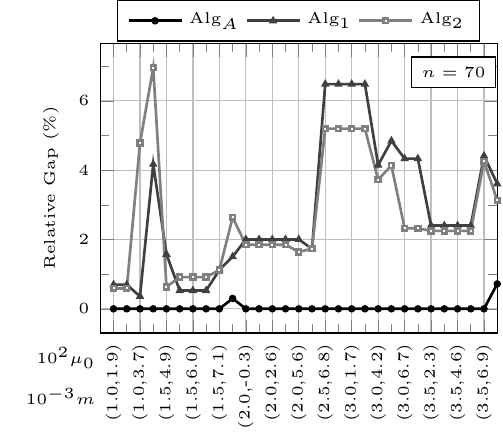}
     \end{subfigure}
  \begin{subfigure}[b]{0.5\linewidth}
    \centering
    \includestandalone[width=\linewidth,  height = \linewidth]{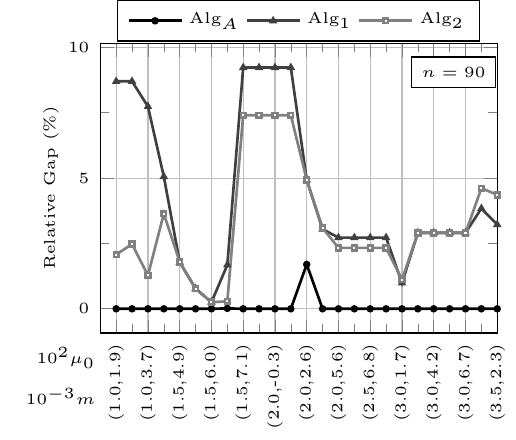}
   \end{subfigure}
     \caption{Comparison of the relative gap between the optimal value
     of the VaR portfolio problem~\eqref{eq:binaryprog} and the lower bounds for~\eqref{eq:binaryprog} provided by Algorithm~\ref{alg:Algorithm-1}, and Algorithms~1 and~2  by~\cite{LarsMU02}.}
\label{fig_LB_2}
\end{figure}

\begin{figure}[!tbh]
    \centering
    \includestandalone[width=\linewidth]{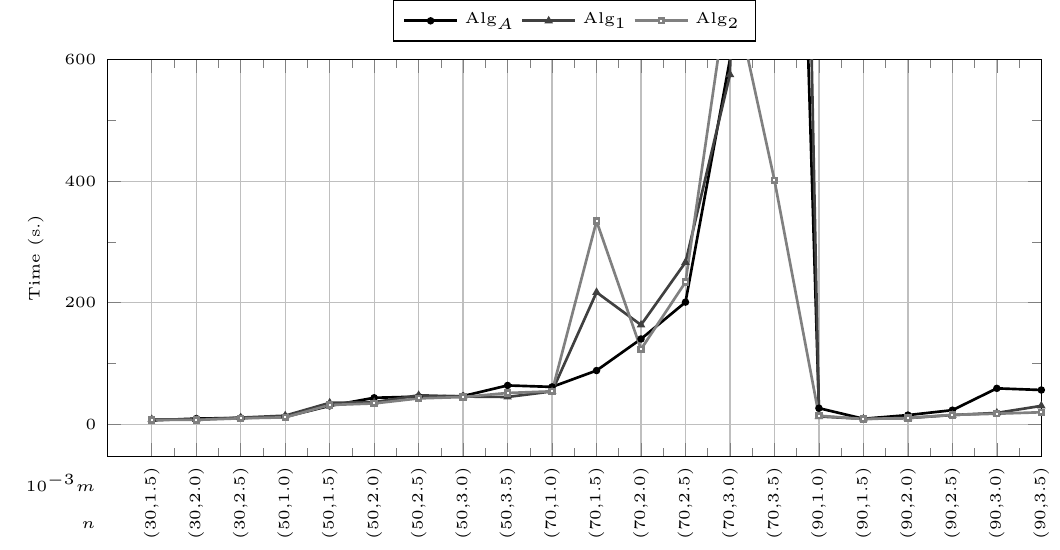}
    \caption{Comparison of average time in seconds needed to run Algorithm~\ref{alg:Algorithm-1}, and Algorithm~1 and~2 in \citep{LarsMU02} to obtain lower bounds for the VaR portfolio problem~\eqref{eq:binaryprog}, for instances with given $n, m$ and different values of $\mu_0$ (cf. Table~\ref{tab:lowbounds}).}
    \label{fig_LB_3}
     \end{figure}

\subsection{Near-optimal VaR portfolio}
In this section, we show that by using Algorithm~\ref{alg:Algorithm-1} and Algorithm~\ref{alg:Algorithm-2}, one
can efficiently compute guaranteed near-optimal solutions for the VaR portfolio problem~\eqref{eq:binaryprog}. For that purpose, to obtain the results described below, we first run Algorithm~\ref{alg:Algorithm-1} with $J_0$ being the first $m_0 = \lceil 2\alpha m \rceil$ samples of the instance. The resulting portfolio $\tilde{x} \in \R^n_+$, VaR lower bound $\tilde{\nu} \in \R$,
and the set $I_0$ (cf., end of  Algorithm~\ref{alg:Algorithm-1}) are then used to initialize  Algorithm~\ref{alg:Algorithm-2},
Also, we set  $\beta = 0.1$, and $\delta = 0.01\tilde{\nu}$. That is, we run
Algorithm~\ref{alg:Algorithm-2} seeking to provide a
$1\%$ near-optimality guarantee for the portfolio $\tilde{x} \in \R^n_+$. In Table~\ref{tab:mainsmall} and Figure~\ref{fig_VaR}, the results of finding a near optimal solution to the VaR portfolio problem using Algorithms~\ref{alg:Algorithm-1} and~\ref{alg:Algorithm-2} versus directly solving the MILP formulation~\eqref{eq:binaryprog} is compared.  For the purpose of brevity, of all the instances considered, Table~\ref{tab:mainsmall} shows, for a particular number of assets $n$ and samples $m$, the instances in which the MILP solver finds the optimal solution of the VaR problem in the shortest and longest time (depending on the value of $\mu_0$). By comparing the columns ``VaR$^*$'' and ``VaR'' in Table~\ref{tab:mainsmall}, it is clear that the lower bound on the VaR portfolio problem~\eqref{eq:binaryprog} it's equal or very close to the optimal value of the VaR portfolio problem~\eqref{eq:binaryprog} (as illustrated also in Figures~\ref{fig_LB_1} and~\ref{fig_LB_2}). Note that these lower bounds are well within the $1\%$ desired tolerance. In
Table~\ref{tab:mainsmall},
$T^*$ is the time taken to solve the MILP formulation~\eqref{eq:binaryprog} of the VaR portfolio problem, and $T$ is the time that is taken to obtain a guaranteed near optimal solution for the VaR portfolio problem using Algorithms~\ref{alg:Algorithm-1} and~\ref{alg:Algorithm-2}. Thus, it is clear from the
$T^*/T$ columns in Table~\ref{tab:mainsmall} that the latter approach is between 1.2 to 46 times faster than directly solving the MILP formulation. On average, the speed up provided by using Algorithms~\ref{alg:Algorithm-1} and~\ref{alg:Algorithm-2} is approximately of 14 times. Given the time is takes to solve some of the instances of the VaR portfolio, this speed up would be crucial to solve practical instances of the VaR portfolio problem. The effect of the speed up provided by Algorithms~\ref{alg:Algorithm-1} and~\ref{alg:Algorithm-2} can be seen graphically in Figure~\ref{fig_VaR}, where the time required by Algorithms~\ref{alg:Algorithm-1} and~\ref{alg:Algorithm-2}, versus the time required to solve the MILP formulation of the VaR optimization problem instances in Table~\ref{tab:mainsmall}, is shown in a semilogarithmic plot.
\begin{figure}[!htb]
    \centering
    \includestandalone[width=0.9\linewidth]{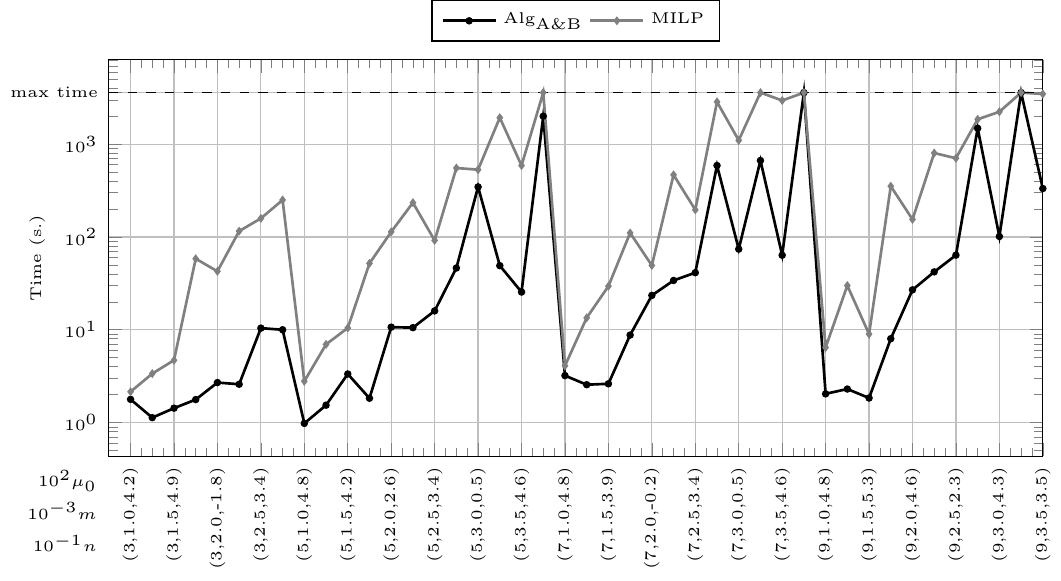}
     \caption{Comparison of the time taken by Algorithm~\eqref{alg:Algorithm-1} and~\eqref{alg:Algorithm-2} vs directly solving the MILP formulation of the VaR portfolio problem~\eqref{eq:binaryprog} for instances with different values of $n$, $m$, and $\mu_0$.}
\label{fig_VaR}
\end{figure}

\begin{table}[!tbh]
\setlength{\abovecaptionskip}{10pt}
{\footnotesize
\begin{center}
\begin{tabular}{cccrcrrcrrrcccccrcrrcrrr}
\toprule
 & && & &       \multicolumn{2}{c}{Full MILP} &&  \multicolumn{3}{c}{Alg.~\ref{alg:Algorithm-1} \&~\ref{alg:Algorithm-2}} &&&&   &&&&     \multicolumn{2}{c}{Full MILP} &&  \multicolumn{3}{c}{Alg.~\ref{alg:Algorithm-1} \&~\ref{alg:Algorithm-2}}\\
 \cmidrule{6-7} \cmidrule{9-11}   \cmidrule{19-20} \cmidrule{22-24}
$n$ &$m$ && \multicolumn{1}{c}{$\mu_0$} && \multicolumn{1}{c}{$\VaR^*$} & \multicolumn{1}{c}{T$^*$}  &&  \multicolumn{1}{c}{$\VaR$} & \multicolumn{1}{c}{T} & \multicolumn{1}{c}{$\rm{T}^*/ \rm{T}$} &&& $n$ & $m$ && \multicolumn{1}{c}{$\mu_0$}  && \multicolumn{1}{c}{$\VaR^*$} & \multicolumn{1}{c}{T$^*$} &&  \multicolumn{1}{c}{$\VaR$} & \multicolumn{1}{c}{T}  & \multicolumn{1}{c}{$\rm{T}^*/T$}\\
\midrule
30	&	1000	&&	0.058	&&	-2.554	&	2.2	&&	-2.560	&	1.78	&	1.2	&~~&~~&	70	&	1500	&&	0.075	&&	-2.002	&	29.5	&&	-2.008	&	2.6	&	11.2\\
30	&	1000	&&	0.042	&&	-1.982	&	3.4	&&	-1.982	&	1.13	&	3.0	&\qquad&\qquad&	70	&	1500	&&	0.039	&&	-1.684	&	110.8	&&	-1.684	&	8.8	&	12.6\\
30	&	1500	&&	0.071	&&	-1.968	&	4.7	&&	-1.968	&	1.43	&	3.3	&\qquad&\qquad&	70	&	2000	&&	0.061	&&	-1.836	&	49.4	&&	-1.836	&	23.5	&	2.1\\
30	&	1500	&&	0.049	&&	-1.702	&	58.4	&&	-1.702	&	1.77	&	32.9	&\qquad&\qquad&	70	&	2000	&&	-0.002	&&	-1.703	&	467.6	&&	-1.703	&	34.0	&	13.7\\
30	&	2000	&&	0.026	&&	-1.721	&	42.7	&&	-1.727	&	2.7	&	15.8	&\qquad&\qquad&	70	&	2500	&&	0.068	&&	-2.121	&	196.4	&&	-2.121	&	41.3	&	4.7\\
30	&	2000	&&	-0.018	&&	-1.721	&	115.8	&&	-1.727	&	2.59	&	44.8	&\qquad&\qquad&	70	&	2500	&&	0.034	&&	-1.876	&	2860.1	&&	-1.876	&	589.6	&	4.8\\
30	&	2500	&&	0.057	&&	-1.951	&	158.7	&&	-1.951	&	10.42	&	15.2	&\qquad&\qquad&	70	&	3000	&&	0.068	&&	-2.024	&	1101.1	&&	-2.024	&	74.1	&	14.8\\
30	&	2500	&&	0.034	&&	-1.951	&	250.6	&&	-1.951	&	10.01	&	25.1	&\qquad&\qquad&	70	&	3000	&&	0.005	&&	\multicolumn{1}{c}{*}	&	\multicolumn{1}{c}{*}	&&	-1.851	&	667.7	&	\multicolumn{1}{c}{$\star$}\\
50	&	1000	&&	0.069	&&	-2.449	&	2.8	&&	-2.449	&	0.98	&	2.8	&\qquad&\qquad&	70	&	3500	&&	0.069	&&	-1.943	&	2960.4	&&	-1.957	&	63.7	&	46.5\\
50	&	1000	&&	0.048	&&	-1.973	&	7.0	&&	-1.973	&	1.54	&	4.5	&\qquad&\qquad&	70	&	3500	&&	0.046	&&	\multicolumn{1}{c}{*}	&	\multicolumn{1}{c}{*}	&&	\multicolumn{1}{c}{*}	&	\multicolumn{1}{c}{*}	&	\multicolumn{1}{c}{*}\\
50	&	1500	&&	0.075	&&	-2.041	&	10.5	&&	-2.050	&	3.34	&	3.1	&\qquad&\qquad&	90	&	1000	&&	0.076	&&	-2.325	&	6.5	&&	-2.339	&	2.0	&	3.1\\
50	&	1500	&&	0.042	&&	-1.699	&	52.0	&&	-1.699	&	1.83	&	28.5	&\qquad&\qquad&	90	&	1000	&&	0.048	&&	-1.906	&	30.0	&&	-1.906	&	2.3	&	13.0\\
50	&	2000	&&	0.011	&&	-1.721	&	113.6	&&	-1.727	&	10.68	&	10.7	&\qquad&\qquad&	90	&	1500	&&	0.084	&&	-2.313	&	9.0	&&	-2.313	&	1.8	&	4.8\\
50	&	2000	&&	0.026	&&	-1.721	&	234.5	&&	-1.727	&	10.55	&	22.2	&\qquad&\qquad&	90	&	1500	&&	0.053	&&	-1.669	&	353.0	&&	-1.669	&	8.0	&	44.0\\
50	&	2500	&&	0.068	&&	-2.164	&	91.9	&&	-2.164	&	15.99	&	5.8	&\qquad&\qquad&	90	&	2000	&&	0.061	&&	-1.825	&	154.8	&&	-1.825	&	27.0	&	5.7\\
50	&	2500	&&	0.034	&&	-1.920	&	554.3	&&	-1.920	&	46.33	&	12.0	&\qquad&\qquad&	90	&	2000	&&	0.046	&&	-1.696	&	803.4	&&	-1.696	&	42.1	&	19.1\\
50	&	3000	&&	0.067	&&	-2.028	&	531.7	&&	-2.004	&	346.59	&	1.5	&\qquad&\qquad&	90	&	2500	&&	0.068	&&	-2.115	&	706.5	&&	-2.115	&	63.8	&	11.1\\
50	&	3000	&&	0.005	&&	-1.887	&	1933.9	&&	-1.887	&	49.15	&	39.4	&\qquad&\qquad&	90	&	2500	&&	0.023	&&	-1.826	&	1855.4	&&	-1.826	&	1483.7	&	1.2\\
50	&	3500	&&	0.069	&&	-1.982	&	589.4	&&	-1.984	&	25.53	&	23.1	&\qquad&\qquad&	90	&	3000	&&	0.068	&&	-2.024	&	2237.0	&&	-2.024	&	101.4	&	22.1\\
50	&	3500	&&	0.046	&&	\multicolumn{1}{c}{*}	&	\multicolumn{1}{c}{*}	&&	-1.849	&	2003.1	&	\multicolumn{1}{c}{$\star$}	&\qquad&\qquad&	90	&	3000	&&	0.043	&&	\multicolumn{1}{c}{*}	&	\multicolumn{1}{c}{*}	&&	\multicolumn{1}{c}{*}	&	\multicolumn{1}{c}{*}	&	\multicolumn{1}{c}{*}\\
70	&	1000	&&	0.069	&&	-2.296	&	4.1	&&	-2.296	&	3.21	&	1.3	&\qquad&\qquad&	90	&	3500	&&	0.069	&&	-1.931	&	3479.2	&&	-1.931	&	333.1	&	10.4\\
70	&	1000	&&	0.048	&&	-1.901	&	13.5	&&	-1.901	&	2.56	&	5.3	&\qquad&\qquad&	90	&	3500	&&	0.035	&&	\multicolumn{1}{c}{*}	&	\multicolumn{1}{c}{*}	&&	\multicolumn{1}{c}{*}	&	\multicolumn{1}{c}{*}	&	\multicolumn{1}{c}{*}\\
\midrule
\multicolumn{6}{l}{Average Speed Up} & & & & &  14.35 & & &\multicolumn{6}{l}{Average Speed Up} & & & & & 13.64\\
\bottomrule
\end{tabular}
\end{center}
}
\caption{Comparison of VaR values and running times of full MILP formulation~\eqref{eq:binaryprog} vs. Algorithms~\ref{alg:Algorithm-1} \&~\ref{alg:Algorithm-2}, for instances of the VaR portfolio problem~\eqref{eq:binaryprog} with different no. of assets ($n$),
no. of samples ($m$), and expected return $\mu_0$. The column~$T^*/T$ shows the speed up obtained with Algorithms~\ref{alg:Algorithm-1} \&~\ref{alg:Algorithm-2} over solving the full MILP formulation. Instances with ``***'' where not solved within the $3600$s. limit time. The ``$\star$'' in column ~$T^*/T$ that ratio cannot be computed.}
\label{tab:mainsmall}
\end{table}

\section{Final Remarks}
\label{s:conclusions}

Thus far, we have only considered portfolio allocation problems in which no short positions are
allowed (i.e., $\mathcal{X} \subseteq \R^n_+$ in~\eqref{nominal}). In practice, none of the
main characteristics of the MILP formulation~\eqref{eq:binaryprog} of the VaR portfolio problem change when considering
portfolios were short positions are allowed (i.e., $\mathcal{X} \subseteq \R^n$ in~\eqref{nominal}).
Clearly, only the choice of the Big-M constant $M$
 is affected by allowing
short positions. However, under the practical assumption that there is $U \in \R_+$  (e.g., due to liquidity)
such that $U \ge \max\{i \in \{1,\dots,n\}: |x_i|\}$, then the $M$ in~\eqref{eq:binaryprog} can be set to
$M > 2U\max\{|\xi^j_i|: i \in \{1,\dots,n\}, j \in \{1,\dots,m\}\}$.

With that said, in this paper, we studied the VaR portfolio selection problem, which is of high relevance in practice, and even in theory, thanks to development of the so-called {\em natural risk statistic} axioms \citep{nocoherent} and the introduction of the concept of {\em elicitability}~\citep[cf.,][]{Bellini13} to classify risk measures. To address the inherently difficult task of solving the VaR portfolio problem, here we propose a tandem of approximation algorithms to produce near-optimal solutions to the VaR portfolio problem. This is done by first using Algorithm~\ref{alg:Algorithm-1} to obtain a good feasible solution for the VaR portfolio problem, and as such, provide a lower bound for the optimal VaR associated with~\eqref{eq:binaryprog}. This algorithm is shown to outperform recent algorithms proposed for this purpose by~\citep{LarsMU02}, based on the iterative solution of appropriate CVaR portfolio problems. Then, in Algorithm~\ref{alg:Algorithm-2}, one aims to prove a \%1 optimallity guarantee for the feasible solution obtained at the end of Algorithm~\ref{alg:Algorithm-1}. The results obtained here, show that using both Algorithm~\ref{alg:Algorithm-1} and~\ref{alg:Algorithm-2} allows to more efficiently solve VaR portfolio problems with up to a hundred assets and thousands of samples, compared to solving the VaR portfolio problem directly with a MILP solver. This results clearly improve the recent results of~\citep{LarsMU02} on lower bounds for the VaR portfolio problem, and the recent results of~\citep[][Sec. 5]{FengWS15} on solving VaR portfolio problems for 25 assets without taking into account the total wealth constraint. Moreover, the proposed algorithms are funded on a study of the alternative formulations of the risk-reward portfolio allocation problem that extends the work done in this area recently by \citet[][Thm. 3]{Krok02}

Finally, we believe that the proposed algorithms can be also applied to solve the broader group of chance constrained optimization problems~\citep[cf.,][]{Sarykalin08}.



\end{document}